\documentclass{article}
\usepackage{amsmath}
  \usepackage{paralist}
  \usepackage{graphics} 
  \usepackage{epsfig} 
 \usepackage[colorlinks=true]{hyperref}
\hypersetup{urlcolor=blue, citecolor=red}

\usepackage[latin1]{inputenc} 
\usepackage{amsfonts} 
\usepackage{amssymb}
\usepackage{amsmath}
\usepackage{graphicx} 
\usepackage{listings} 
\usepackage{amsthm} 
\usepackage{color}

\usepackage{float}
\usepackage{tikz}
\usepackage{todonotes}
\usetikzlibrary{positioning,patterns,shapes,calc,decorations.pathmorphing,patterns,backgrounds,
decorations.pathreplacing,trees}

\newcommand{\N}{\mathbb{N}}
\newcommand{\Z}{\mathbb{Z}}

\newcommand{\cC}{\mathcal{C}}
\newcommand{\URi}{\mathcal{UR}^-_r}
\newcommand{\ULi}{\mathcal{UL}^-_r}
\newcommand{\BRi}{\mathcal{BR}^-_r}
\newcommand{\BLi}{\mathcal{BL}^-_r}
\newcommand{\URe}{\mathcal{UR}^+_r}
\newcommand{\ULe}{\mathcal{UL}^+_r}
\newcommand{\BRe}{\mathcal{BR}^+_r}
\newcommand{\BLe}{\mathcal{BL}^+_r}
\newcommand{\bigvectwo}[2]{\begin{pmatrix}#1\\#2\end{pmatrix}}

\newcommand{\bigmattwo}[4]{\begin{pmatrix}#1&#2\\#3&#4\end{pmatrix}}
\newcommand{\vectwo}{\bigvectwo}

\newtheorem{theorem}{Theorem}[section]
\newtheorem{corollary}{Corollary}

\newtheorem{lemma}[theorem]{Lemma}

\theoremstyle{definition}
\newtheorem{definition}[theorem]{Definition}
\newtheorem{remark}{Remark}

\newtheorem{example}[theorem]{Example}

\begin{document}

\title{Identifying Codes of Degree 4 Cayley Graphs over Abelian Groups}
\author{Crist\'{o}bal Camarero  and Carmen Mart\'{\i}nez  and Ram\'{o}n Beivide}
\date{}

\maketitle




\bigskip


\begin{abstract}
In this paper a wide family of identifying codes over regular Cayley graphs of degree four which are built over finite Abelian groups is presented. Some of the codes in this construction are also perfect. The graphs considered include some well-known graphs such as tori, twisted tori and Kronecker products of two cycles. Therefore, the codes can be used for identification in these graphs. Finally, an example of how these codes can be applied for adaptive identification over these graphs is presented.
\end{abstract}

\section{Introduction}

Systems are continuously growing in complexity, thus becoming composed of more and more basic elements. This originates many important problems which must
be solved in order to use the systems effectively. Among these problems are the parallelization of algorithms to exploit the full system, the development of networks efficient in cost and in performance, and the automatic detection of failures, and consequently, reconfiguration and recovery of the system.\par

In this paper the problem of the detection of faults will be explored, with principal aim at multiprocessor systems, in which
the basic elements are computers. These computers are connected between them with some topology which can be modeled by a (usually undirected) graph.
The graphs under consideration will be regular Cayley graphs of degree four, which englobe a multitude of well-known interconnection networks.\par

The major mechanisms for fault-diagnosis are the PMC model by Preparata, Mezte and Chien \cite{Preparata}
and the comparison model by Malek \cite{Malek}. In the PMC model, each unit tests each neighbor to see if it is faulty
and the tests performed by faulty units are unreliable. Then, all the information is used globally to determinate the faulty vertices.
A bit of information is generated for each edge and direction, which is a considerable volume of information. Later, Karpovsky, Chakrabarty and Levitinin in \cite{Karpovsky} proposed identifying codes with the aim of reducing the necessary information for fault diagnosis.\par

Identifying codes can be seen as a subset of vertices of a graph which allows the identification of single vertices under some conditions. When identifying codes apply for diagnosis, for each codeword a bit of information is generated, which is significant less information than the number of edges, as required in the PMC model. Since minimizing the global information is equivalent to minimizing the density of the code, a lot of work has been done to find codes with minimum density over different infinite graphs. There are hundreds of papers \cite{Lobsteinweb} on identifying codes, many of them looking for the ones which attain the minimal density bound over different graphs.
However, there are not many constructive methods for finding wide families of these codes over graphs which correspond to topological models for multiprocessor systems.
Moreover, technical aspects of the diagnosis process using identifying codes have not been extensively considered. An algorithmic approach to the adaptive identification in multiprocessor systems has been considered, for example in \cite{Benhaimadaptive} and \cite{Benhaimadaptive2}.\par

In this paper a constructive method for finding a wide family of identifying codes over degree four Cayley graphs over finite Abelian groups is provided. The codes are built using subgroups, which implies that the codewords are evenly distributed. The identifying radius as well as the covering radius are characterized in terms of the generators of the subgroup. Moreover, such codes in some cases are also perfect codes, but for a different radius than the identifying one. Hence, in Section \ref{sec:Cayley} the graphs that are going to be consider in the paper are defined. In Section \ref{sec:identifyingdefinitions} some basic definitions of identification are recalled and a few new ones that will be needed to make the general construction are established. In Section \ref{sec:construccion} the identifying radius for a family of group codes is determined. Moreover, both the case in which the codes constructed are not only identifying but also perfect, and the density of the construction are considered. Finally in Section \ref{sec:algoritmo} an example of how to apply these codes for adaptive identification in faulty environments modeled by Cayley graphs is shown. Some of the results presented in this paper were first announced in \cite{3ICMTA}.

\section{Cayley Graphs over Two Dimensional Lattices} \label{sec:Cayley}

In this section some definitions and results concerning two dimensional lattices are introduced. Then, degree four regular Cayley graphs over these lattices are defined for obtaining a metric space in which identifying and perfect codes are possible to be constructed.\par

Given two linearly independent integer vectors $m_1, m_2 \in \mathbb{Z}^2$ let us consider $$ \langle m_1 , m_2 \rangle   = \{\lambda_1 m_1 + \lambda_2 m_2 \ | \ \lambda_i \in \mathbb{Z}, i=1,2 \},$$ the integer lattice (or additive group) generated by vectors $m_1 , m_2$. As usual, the quotient group can be defined as: $$\mathbb{Z}^{2}/\langle m_1 , m_2 \rangle  \mathbb{Z}^{2} = \{ v + \langle m_1 , m_2 \rangle   \ | \ v \in \mathbb{Z}^{2} \}.$$ The following theorem will be used (for the proof see \cite{Fiolmultidimensional}):

\begin{theorem}\label{7:theo:cardinal-grupo} $\mathbb{Z}^{2}/\langle m_1 , m_2 \rangle \mathbb{Z}^{2}$ has $| \det(M) |$ elements, where $M$ is the matrix whose columns are the linearly independent vectors $\{ m_1 , m_2 \}$.
\end{theorem}

Following the notation in \cite{Fiolmultidimensional}, two vectors $v_1, v_2$ in the same coset of the quotient group $\mathbb{Z}^{2}/\langle m_1 ,m_2 \rangle  \mathbb{Z}^{2}$ will be denoted as $v_1 \equiv v_2 \pmod{M}$. The quotient group $\mathbb{Z}^{2}/\langle m_1 ,m_2 \rangle  \mathbb{Z}^{2}$ will be denoted as $\mathbb{Z}^{2}/M\mathbb{Z}^{2}$ and a subgroup $\langle t_1 ,t_2 \rangle   \subseteq \mathbb{Z}^{2}/\langle m_1 ,m_2 \rangle  \mathbb{Z}^{2}$ as $T + \mathbb{Z}^{2}/M\mathbb{Z}^{2}$, where $T$ is the matrix whose columns are $t_1$ and $t_2$. As usual, $\mathcal{M}_{2 \times 2}(\mathbb{Z})$ denotes the ring of the square integer matrices of size two. In the remainder of the paper, it will be assumed $M \in \mathcal{M}_{2 \times 2}(\mathbb{Z})$ to be non-singular.\par

The Cayley graphs considered in the paper are defined over two-dimensional lattices. As usual, a Cayley graph is defined:

\begin{definition}
The \emph{Cayley graph} over a group $\Gamma$ and a set $A\subset \Gamma$ is the graph $Cay(\Gamma;A)$ with vertex set $V=\Gamma$ and edges
	$$E=\{(v,v+g) \mid v\in \Gamma,g\in A\}$$
The set $A$ is called the \emph{adjacency set}.
\end{definition}

In this paper Cayley graphs over the quotient rings defined above are going to be considered. Therefore, given $M \in \mathcal{M}_{2 \times 2}(\mathbb{Z})$ a non-singular matrix and $A = \{a_1 , a_2\}$ a subset of vectors of $\mathbb{Z}^{2}$, the Cayley graph over the group $\mathbb{Z}^{2}/M \mathbb{Z}^2$ with adjacency $A \subset \mathbb{Z}^{2}$ will be denoted by $\mathcal{G}(M; A)$. Therefore, every vertex $u \in \mathbb{Z}^{2}/M \mathbb{Z}^2$ of $\mathcal{G}(M; A)$ is adjacent to $u \pm A \pmod{M}$. Moreover, for convenience, $E = \{e_1, e_2\}$ will denote the orthonormal two-dimensional basis, where $e_1 = \vectwo{0}{1}$ and $e_2 = \vectwo{1}{0}$. Without loss of generality, it can be assumed that any $\mathcal{G}(M; A)$ has as set of jumps $A = E$. Therefore, $\mathcal{G}(M) = \mathcal{G}(M; E)$. This way, any $\mathcal{G}(M)$ is locally a two-dimensional orthogonal mesh, with wrap-around edges completing the peripheral adjacency. The graphs considered in the present paper constitute the two dimensional case of the general family of multidimensional circulants defined in \cite{Fiolmultidimensional}. Moreover, as it is stated in the following result, these are in fact all the Cayley graphs of degree four over finite Abelian groups.\par

\begin{lemma} \cite{Bermond} \label{lemma:Bermond} Let $\Gamma$ be a finite Abelian group, $A=\{\pm a,\pm b\}\subseteq \Gamma$ and consider $Cay(\Gamma;A)$. Then, there exists $M \in \mathcal{M}_{2 \times 2}(\mathbb{Z})$ such that
$$Cay(\Gamma;A) \cong \mathcal{G}(M).$$
\end{lemma}

Many topological models for interconnection networks are in fact Cayley graphs as the ones that are being considered in this paper. For example, tori, twisted tori \cite{Sequin} and circulant graphs of degree four such as the ones in \cite{Coppersmith} and \cite{Fiolnetworks}. As it was proved in \cite{ISIT10}, also the Kronecker product of two cycles is a particular case of $\mathcal{G}(M)$.

\begin{example} Let $M = \begin{pmatrix} 12 & 5 \\ 1 & 5 \end{pmatrix}$. In Figure \ref{7:fig:ej1} a graphical representation of the graph $\mathcal{G}(M)$ in which
the representative at minimum distance from zero of every vertex is shown.

\begin{figure}
   \centering
      \includegraphics[width=.6\columnwidth]{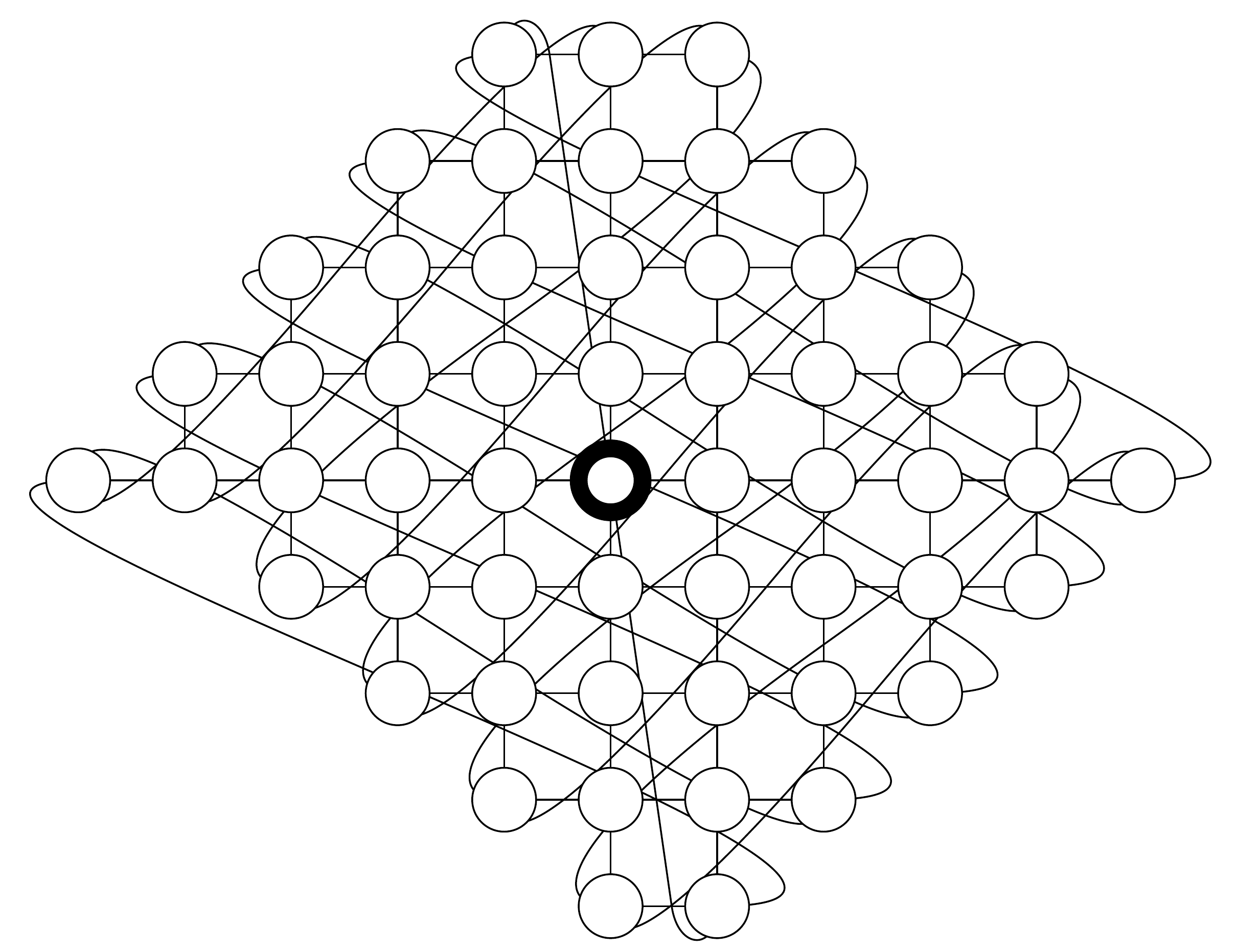}
      \caption{$\mathcal{G}(M)$ represented in minimum distances.}
      \label{7:fig:ej1}
\end{figure}
\end{example}

Finally, if $d_{M}(v_1 , v_2)$ denotes the distance between vertices $v_1$ and $v_2 $ of $\mathcal{G}(M)$, then it can be calculated as: $$d_{M}(v_1 , v_2) = |\lambda_1| + |\lambda_2|,$$ with $\lambda_1, \lambda_2 \in \mathbb{Z}$ such that $$v_1 - v_2 \equiv \lambda_1 e_1 + \lambda_2 e_2 \pmod{M}$$ and $|\lambda_1| + |\lambda_2| \neq 0$ being minimum.

\section{Previous Definitions on Identification} \label{sec:identifyingdefinitions}

The \emph{ball} of radius $r$ centered at $v\in V$ is defined as: $$B_r(v)=\{w\in V \mid  d(v,w)\leq r\}.$$

Then, it is also said that a vertex $v$ \emph{$r$-covers} $w$, for any $w \in B_r(v)$. A \emph{code} $\cC$ is a nonempty set of vertices, whose elements are called \emph{codewords}. Given a vertex $v\in V$ the set of codewords which $r$-cover $v$ are denoted as: $$K_r(v)=\cC \cap
B_r(v).$$ Moreover, two vertices $v,w\in V$ are said to be \emph{$r$-separated} if $K_r(v)\neq K_r(w)$.

\begin{definition}\cite{Karpovsky} A code $\cC$ of a graph $G$ is said
\emph{$r$-identifying} if the sets $K_r(v)$ are all nonempty and different, that is,
every pair of vertices are $r$-separated and each vertex is $r$-covered by
at least one codeword.
\end{definition}

Although the aim of this paper is to construct identifying codes over $\mathcal{G}(M)$ graphs, the problem over the infinite mesh will be considered before. Previous papers dealing with the problem of identification over the infinite mesh are, for example, \cite{Charonsmallradius}, \cite{Honkala} and \cite{Charon}. Therefore, let us define the infinite mesh as the Cayley graph $Cay(\Z^2,E)$.
This graph will be denoted as $\mathbb{Z}^{2}$ whenever there is no possibility of confusion. Also,
the elements in $\Z^2$ will be in column form, for convenience. The following new definitions will be also needed.

\begin{definition} The following subsets of $\Z^2$ are defined:
\begin{align*}
\URi&=\left \{\vectwo{x}{y} \in \mathbb{Z}^{2} \mid 0\leq x,y \wedge |x|+|y|=r \right \}\\
\ULi&=\left \{\vectwo{x}{y} \in \mathbb{Z}^{2}\mid x\leq 0\leq y \wedge |x|+|y|=r \right \}\\
\BRi&=\left \{\vectwo{x}{y} \in \mathbb{Z}^{2}\mid y\leq 0\leq x \wedge |x|+|y|=r \right \}\\
\BLi&=\left \{\vectwo{x}{y} \in \mathbb{Z}^{2}\mid x,y\leq 0 \wedge |x|+|y|=r \right \}\\
\URe&=\left \{\vectwo{x}{y} \in \mathbb{Z}^{2}\mid 0 <  x,y \wedge |x|+|y|=r+1 \right \}\\
\ULe&=\left \{\vectwo{x}{y} \in \mathbb{Z}^{2}\mid x < 0 < y \wedge |x|+|y|=r+1 \right \}\\
\BRe&=\left \{\vectwo{x}{y} \in \mathbb{Z}^{2}\mid y < 0 < x \wedge |x|+|y|=r+1 \right \}\\
\BLe&=\left \{\vectwo{x}{y} \in \mathbb{Z}^{2}\mid x,y < 0 \wedge |x|+|y|=r+1 \right \}\\
\end{align*}
The translated sets are denoted as $\URi(v)=v+\URi$ and analogously for the other sets. A graphical representation of these sets is shown in Figure \ref{fig:conjuntos}.
\end{definition}

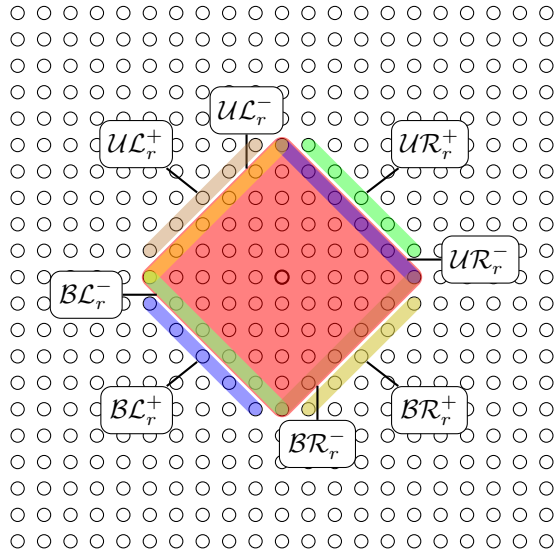
\begin{figure}
	\begin{center}
	\begin{tikzpicture}[x=10pt,y=10pt,
		every pin/.style={opacity=1,fill=white,draw},
		every pin edge/.style={opacity=1,thick},
		pinner/.style={minimum width=0pt,inner sep=0pt},
		]
	\foreach \x in{0,1,...,20}
		\foreach \y in {0,1,...,20}
			\node[draw,circle,minimum width=5pt,inner sep=0pt]
				(nodo\x_\y) at (\x,\y) {};
	\path(nodo10_10) node[draw,circle,thick,minimum width=5pt,inner sep=0pt] {};
	\def\k{2pt}
	\fill[fill=red,opacity=0.5,rounded corners]
		($(nodo10_5.south)+(0,-\k)$) -- ($(nodo15_10.east)+(\k,0)$)
		--($(nodo10_15.north)+(0,\k)$) -- ($(nodo5_10.west)+(-\k,0)$) --cycle;
	\fill[fill=blue,opacity=.4,rounded corners]
		($(nodo10_15.north east)+(-\k,\k)$)
		--($(nodo15_10.north east)+(\k,-\k)$)
		node[very near end,pinner,pin=right:$\URi$]{}
		--($(nodo15_10.south west)+(\k,-\k)$)
		--($(nodo10_15.south west)+(-\k,\k)$)
		--cycle;
	\fill[fill=blue,opacity=.4,rounded corners]
		($(nodo5_9.north east)+(-\k,\k)$)
		--($(nodo9_5.north east)+(\k,-\k)$)
		--($(nodo9_5.south west)+(\k,-\k)$)
		--($(nodo5_9.south west)+(-\k,\k)$)
		node[pinner,midway,pin=below left:$\BLe$]{}
		--cycle;
	\fill[fill=green!50,opacity=.6,rounded corners]
		($(nodo5_10.north east)+(-\k,\k)$)
		--($(nodo10_5.north east)+(\k,-\k)$)
		--($(nodo10_5.south west)+(\k,-\k)$)
		--($(nodo5_10.south west)+(-\k,\k)$)
		node[very near end,pinner,pin=left:$\BLi$]{}
		--cycle;
	\fill[fill=green,opacity=.4,rounded corners]
		($(nodo11_15.north east)+(-\k,\k)$)
		--($(nodo15_11.north east)+(\k,-\k)$)
		node[pinner,midway,pin=above right:$\URe$]{}
		--($(nodo15_11.south west)+(\k,-\k)$)
		--($(nodo11_15.south west)+(-\k,\k)$)
		--cycle;
	\fill[fill=yellow,opacity=.4,rounded corners]
		($(nodo10_15.north west)+(\k,\k)$)
		--($(nodo5_10.north west)+(-\k,-\k)$)
		node[near start,pinner,pin=above:$\ULi$]{}
		--($(nodo5_10.south east)+(-\k,-\k)$)
		--($(nodo10_15.south east)+(\k,\k)$)
		--cycle;
	\fill[fill=yellow!80!black,opacity=.6,rounded corners]
		($(nodo15_9.north west)+(\k,\k)$)
		--($(nodo11_5.north west)+(-\k,-\k)$)
		--($(nodo11_5.south east)+(-\k,-\k)$)
		--($(nodo15_9.south east)+(\k,\k)$)
		node[midway,pinner,pin=below right:$\BRe$]{}
		--cycle;
	\fill[fill=brown,opacity=.4,rounded corners]
		($(nodo15_10.north west)+(\k,\k)$)
		--($(nodo10_5.north west)+(-\k,-\k)$)
		--($(nodo10_5.south east)+(-\k,-\k)$)
		--($(nodo15_10.south east)+(\k,\k)$)
		node[near start,pinner,pin=below:$\BRi$]{}
		--cycle;
	\fill[fill=brown,opacity=.4,rounded corners]
		($(nodo9_15.north west)+(\k,\k)$)
		--($(nodo5_11.north west)+(-\k,-\k)$)
		node[midway,pinner,pin=above left:$\ULe$]{}
		--($(nodo5_11.south east)+(-\k,-\k)$)
		--($(nodo9_15.south east)+(\k,\k)$)
		--cycle;
	\end{tikzpicture}
	\end{center}
	\caption{Subsets that divide the periphery of the $r$-ball.}
	\label{fig:conjuntos}
\end{figure}

The next definition and lemma will be useful to prove the identifiability of the codes proposed in this paper.

\begin{definition} Let $\cC$ be a code of a graph $\mathcal G(M)$ and $v$ be a vertex.
The ball $B_r(v)$ is said to be \emph{$r$-fixed} by $\cC$ if there exists at least a codeword in each one of the
next sets:
$$
\URi(v)\cup\BLe(v), \
\ULi(v)\cup\BRe(v), \
\BRi(v)\cup\ULe(v), \
\BLi(v)\cup\URe(v)
$$
\end{definition}

\begin{remark} Although having the ball in every vertex $r$-fixed is not a
sufficient condition to be $r$-identifying (as the example in Figure
\ref{fig:fixednotidentifying}), this property will be used as a first step in
the proofs, since it simplifies them.
\end{remark}

\begin{figure}
    \begin{center}
    \includegraphics[width=.6\columnwidth]{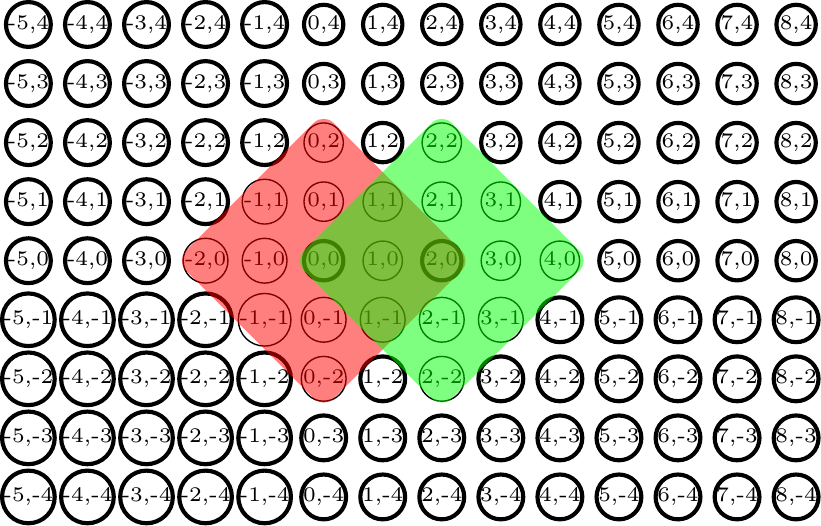}
	\end{center}
    \caption{
		Every ball is $2$-fixed, but it is not $2$-identifying.
		}
    \label{fig:fixednotidentifying}
\end{figure}

\begin{lemma}If $G$ is a Cayley graph over a group $\Gamma$
and $\cC$ is a subgroup of $\Gamma$, then $\cC$ is a code such that for all $c\in\cC$, $K_r(v+c)=c+K_r(v).$
\end{lemma}

\section{The Construction of Identifying Codes} \label{sec:construccion}

In this section a wide family of identifying codes
over the infinite mesh will be constructed. Then, by making the suitable modular operation by a matrix $M$, these
codes will also be identifying over a graph $\mathcal{G}(M)$.\par

Let $t$ and $d$ be two positive integers. Let us consider the code $$\cC = \langle
\vectwo{t}{t+d}, \vectwo{-(t+d)}{t}\rangle $$ of the infinite mesh. Also, $C =\bigmattwo{t}{-(t+d)}{t+d}{t}$ will denote the matrix associated to the group code.
If $N$ denotes the cardinal of $\Z^2/\cC \Z^2$, it is obtained that
$N=\det(C)=t^2+(t+d)^2$. Note that, as it was proved in \cite{IEEE_TC}, the covering radius of the code $\cC$ is $t+d$.\par

In the next subsections it will be determined the radius $r$ such that
$\mathcal{C}$ is $r$-identifying over $\mathbb{Z}^2$.
Depending on $t$ and $d$ there are three cases:
\begin{enumerate}
\item if $\gcd(2t,d)=1$ then $\cC$ is $r$-identifying for
	$r=t^2+dt+(d^2-1)/2$,
\item if $\gcd(2t,d)=2$ then $\cC$ is $r$-identifying for
	$r=t^2+dt+d^2/2$,
\item if $\gcd(2t,d)\geq 3$ then $\cC$ is not $r$-identifying for any $r$.
\end{enumerate}

As a consequence, the covering radius is quite smaller than the identifying. This will be beneficial for the identification process as
it will be remarked in Section \ref{sec:algoritmo}. Moreover, sufficient
conditions for $\cC$ to be also $r$-identifying over $\mathcal{G}(M)$ will be given. Finally,
the density of the construction will be considered.

\subsection{First Case: $\gcd(2t, d)=1$}

Let $t$ and $d$ be two positive integers such that $\gcd(2t,d)=1$. Let us consider the code $\cC = \langle \vectwo{t}{t+d}, \vectwo{-(t+d)}{t}\rangle$ over $\mathbb{Z}^2$ and let us denote by $N$
the volume of the fundamental parallelogram of $\mathcal{C}$, that is, $N = \det(C)= t^2+(t+d)^2$. The following technical result can be proved:

\begin{lemma}\label{lemma:codewords-impar}If $\vectwo{x}{y}\in\cC$ then for every $m,n\in\Z,\ \vectwo{x+mN}{y+nN}\in\cC$.
\end{lemma}

\begin{proof} The result is obtained by

$$\bigmattwo{t}{-(t+d)}{t+d}{t}
	\bigvectwo{t+d}{t}=\bigvectwo{0}{N},$$
and
$$\bigmattwo{t}{-(t+d)}{t+d}{t}
	\bigvectwo{t}{-t-d}=\bigvectwo{N}{0}.$$
\end{proof}

\begin{lemma}\label{lemma:fijas-impar} For all $v\in \mathbb{Z}^{2}$, $B_r(v)$ is $r$-fixed by $\cC = \langle
\vectwo{t}{t+d}, \vectwo{-(t+d)}{t}\rangle $, where $r =t^2+dt+\frac{d^2-1}{2}$
and $\gcd(2t, d)=1$.
\end{lemma}

\begin{proof}
It will proved that for each $v \in \mathbb{Z}^{2}$ there exists a codeword $c$ in the set
$\URi(v)\cup\BLe(v)$. The other three codewords can be obtained by making the
suitable rotations. Note that, if $c\in\URi(v)\cup\BLe(v)$ and
$R=\begin{pmatrix}0&-1\\1&0\\\end{pmatrix}$ then $Rc$ is codeword in
$\ULi(Rv)\cup\BRe(Rv)$, $R^2c$ is codeword in $\BLi(R^2v)\cup\URe(R^2v)$ and
$R^3c$ is codeword in $\BRi(R^3v)\cup\ULe(R^3v)$.\par

Since $c$ must be in $\cC$, there must exist $a,b \in \mathbb{Z}$ such that
$c=\vectwo{at-b(t+d)}{a(t+d)+bt}$. Now, if $v = \vectwo{v_1}{v_2}$, let us
denote $\vectwo{x}{y}=c-v.$ Then, the following system of Diophantine
equations is obtained:
$$
\left\{
\begin{array}{l}
x=at-b(t+d)-v_1\\
y=a(t+d)+bt-v_2
\end{array}\right.
$$

Therefore, $x+y=a(2t+d)-bd-v_1-v_2$. By hypothesis, $\gcd(2t+d,d)=\gcd(2t,d)=1$, which implies that there exist $a_0,b_0 \in \mathbb{Z}$ such that
$a_0(2t+d)-b_0d=r+v_1+v_2.$ Thus, there also exist $x_0,y_0 \in \mathbb{Z}$ such that $x_0+y_0=r$.
The set of solutions is $\{a'=a_0+d\lambda, b'=b_0+(2t+d)\lambda \ | \ \lambda \in \mathbb{Z} \}$.
Therefore, let $x'=a_0t-b_0(t+d)+td\lambda-(t+d)(2t+d)\lambda-v_1
=a_0t-b_0(t+d)-v_1-N\lambda=x_0-\lambda N$ and in the same way let $y'=y_0+\lambda N$.\par

Now, it is set $\lambda$ such that $0\leq x' < N=2r+1$, so $0\leq x'\leq 2r$.
Two different cases are obtained:

\begin{itemize}
\item If $0\leq x'\leq r$ then it is obtained that $0\leq y'\leq r$. Therefore,
	taking $\vectwo{x}{y}=\vectwo{x'}{y'}$ implies that $c=v+\vectwo{x}{y}$ is in $\URi(v)$.
\item Otherwise, it is obtained that $r < x'\leq 2r$ so $-r\leq y' < 0$. Therefore,
	taking $\vectwo{x}{y}=\vectwo{x'-N}{y'}$ (which is a codeword by Lemma \ref{lemma:codewords-impar}) implies that $x+y=r-N=-r-1$ and
	$-r-1 < x\leq -1$, that is $-r\leq x < 0$, thus obtaining that $c=v+\vectwo{x}{y}$ is a codeword in $\BLe(v)$.
\end{itemize}
\end{proof}

\begin{theorem}\label{thm:identifying} Let $t, d$ be positive integers such that $\gcd(2t,d)=1$ and $r=t^2+dt+\frac{d^2-1}{2}$. Then, $\cC = \langle \vectwo{t}{t+d}, \vectwo{-(t+d)}{t}\rangle  $ is a
$r$-identifying code in $\mathbb{Z}^2$.
\end{theorem}

\begin{proof} By the previous lemma, all $r$-balls are $r$-fixed, and hence, $\cC$
$r$-dominates $\mathbb{Z}^2$. However, this is not enough for
being an $r$-identifying code. Therefore, it will be next proved that any $v\in
\mathbb{Z}^{2}$ is $r$-separated from $w=v+\vectwo{x}{y}$, for any
$\vectwo{x}{y} \in \mathbb{Z}^{2}$. Note that, if $|x|=|y|$ this is
straightforward since the balls are $r$-fixed.
Then, as each 90-degree rotation leaves $\cC$ invariant, it can be assumed without loss of generality that $x>y\geq 0$.\par

If there is a codeword in $\ULi(v)$ then it is straightforward that they are $r$-separated. Otherwise, there is a codeword $c_1\in \cC\cap\BRe(v)$ which is not in $K_r(v)$. If $c_1$ is in $K_r(w)$ the proof has finished. As $c_1\in \BRe(v)$ it can be expressed as $c_1=v+\vectwo{a}{a-r-1}$ with $1\leq a\leq r$.
Then,
	$$c_1+\vectwo{-(t+d)}{t}=v+\vectwo{a-(t+d)}{a+t-r-1}$$
	is a codeword which is further from $K_r(w)$ than $c_1$.
And when $a>\frac{d-1}{2}$, it is in $K_r(v)$.

If the same than in the previous case is done, but for $\BLi(v)\cup\URe(v)$, it is obtained that $c_2=v+\vectwo{b}{r-b+1}\in\URe(v)$.
Then,
	$$c_2+\vectwo{-(t+d)}{t}-\vectwo{t}{t+d}=c_2+\vectwo{-2t-d}{-d}=v+\vectwo{b-2t-d}{r-b-d+1}$$
which is further from $K_r(w)$ than $c_2$
and when $b>t$ it is in $K_r(v)$.\par

To finish, it has to be proved that the case when both codewords $c_1,c_2$ are in the left ranges never happens.
For that, let us assume that $1\leq a\leq\frac{d-1}{2}$ and $1\leq b\leq t$, with $c_1,c_2\in\cC$.
Thus, $c=c_1-c_2=\vectwo{a-b}{a+b-1-N}=\vectwo{x'}{y'}$ is in the rectangle defined by:

$$
	\{1-t\leq x'< \frac{d-1}{2}-1 < d,\ 1-N\leq y'\leq t+\frac{d-1}{2}-1-N < t+d-N\}
$$
which does not include the point $\vectwo{0}{-N}$.
This rectangle is a subset of
\begin{multline*}
	\{1-2t\leq x'\leq -t, -N\leq y'\leq -N+t-1\} \cup\\
	\{ 1-t\leq x'\leq d, -N\leq y'\leq -N+t+d-1\}
\end{multline*}
which is  a set of representatives of $\Z^2/\cC \Z^2$ (see \cite{Jordan})
and contains $\vectwo{0}{-N}$. As a consequence, all its elements are different.
Therefore, $c$ is not congruent with $\vectwo{0}{0}$ and thus it is not a codeword, which is a contradiction.\end{proof}

Once a family of identifying codes over $\mathbb{Z}^2$ has been constructed, new conditions have to be added in order to get an identifying code also over $\mathcal{G}(M)$. The most obvious condition is that the code must be periodic. Then, in the case of these lattice codes $M$ has to be a multiple of the matrix $C$ associated to the lattice code. With this it is obtined that the code is locally the same than in the plane, thus keeping its density. The second condition is that any ball cannot overlap with itself. These conditions are stated in the following result.

\begin{corollary}\label{col:Lidentifying} Let $t, d$ be two positive integers and  $M \in \mathcal{M}_{2 \times 2}(\mathbb{Z})$. Let us consider the graph $\mathcal{G}(M)$ and its code $\mathcal{C} = \langle \vectwo{t}{t+d}, \vectwo{-(t+d)}{t}\rangle $. If there exists $Q \in \mathcal{M}_{2 \times 2}(\mathbb{Z})$ such that $M = CQ$ and $\min \{ d(0, M\gamma) \ | \ 0 \neq \gamma \in \mathbb{Z}^2 \} > 2r+1$ then $\cC$ is a $r$-identifying code in $\mathcal{G}(M)$ for $r =t^2+dt+\frac{d^2-1}{2}$.
\end{corollary}

\begin{remark} Note that $\min \{ d(0, M\gamma) \ | \ 0 \neq \gamma \in \mathbb{Z}^2 \}$ is the length of the shortest non-trivial cycle in $\mathcal{G}(M)$. \end{remark}

In \cite{ISIT10} it was shown that, under some conditions for matrix $M$, $\cC=\allowbreak\langle\vectwo{t}{t+1},\allowbreak\vectwo{-t-1}{t}\rangle$ is a perfect $t$-error correcting code over the
graph $\mathcal{G}(M)$. Moreover, this is the only perfect code that can be constructed over that graph, up to symmetries. Next, this
result about perfect codes and the previous one obtained in this section dealing with identifying codes are summarized in the following corollary.

\begin{corollary}\label{col:perfecto} Let $t$ be a positive integer and $M \in \mathcal{M}_{2 \times 2}(\mathbb{Z})$. Let us consider the graph $\mathcal{G}(M)$ and its code $\cC = \langle \vectwo{t}{t+1}, \vectwo{-(t+1)}{t}\rangle $. If there exists $Q \in \mathcal{M}_{2 \times 2}(\mathbb{Z})$ such that $M = CQ$ and $\min \{ d(0, M\gamma) \ | \ 0 \neq \gamma \in \mathbb{Z}^2 \} > 2t(t+1)+1$ then $\cC$ is a $t(t+1)$-identifying and $t$-perfect code in $\mathcal{G}(M)$.
\end{corollary}

\begin{example} Let us consider $t=1$ and the matrix $M = \begin{pmatrix}25&0\\0&25\end{pmatrix}.$ Then, the code $\cC = \langle \vectwo{1}{2}, \vectwo{-2}{1} \rangle  $ is 2-identifying in $\mathcal{G}(M)$ which, as it can be seen in Figure \ref{fig:2identifT}, is a square torus graph of side 25.

\begin{figure}
    \begin{center}
    \includegraphics[width=.6\columnwidth]{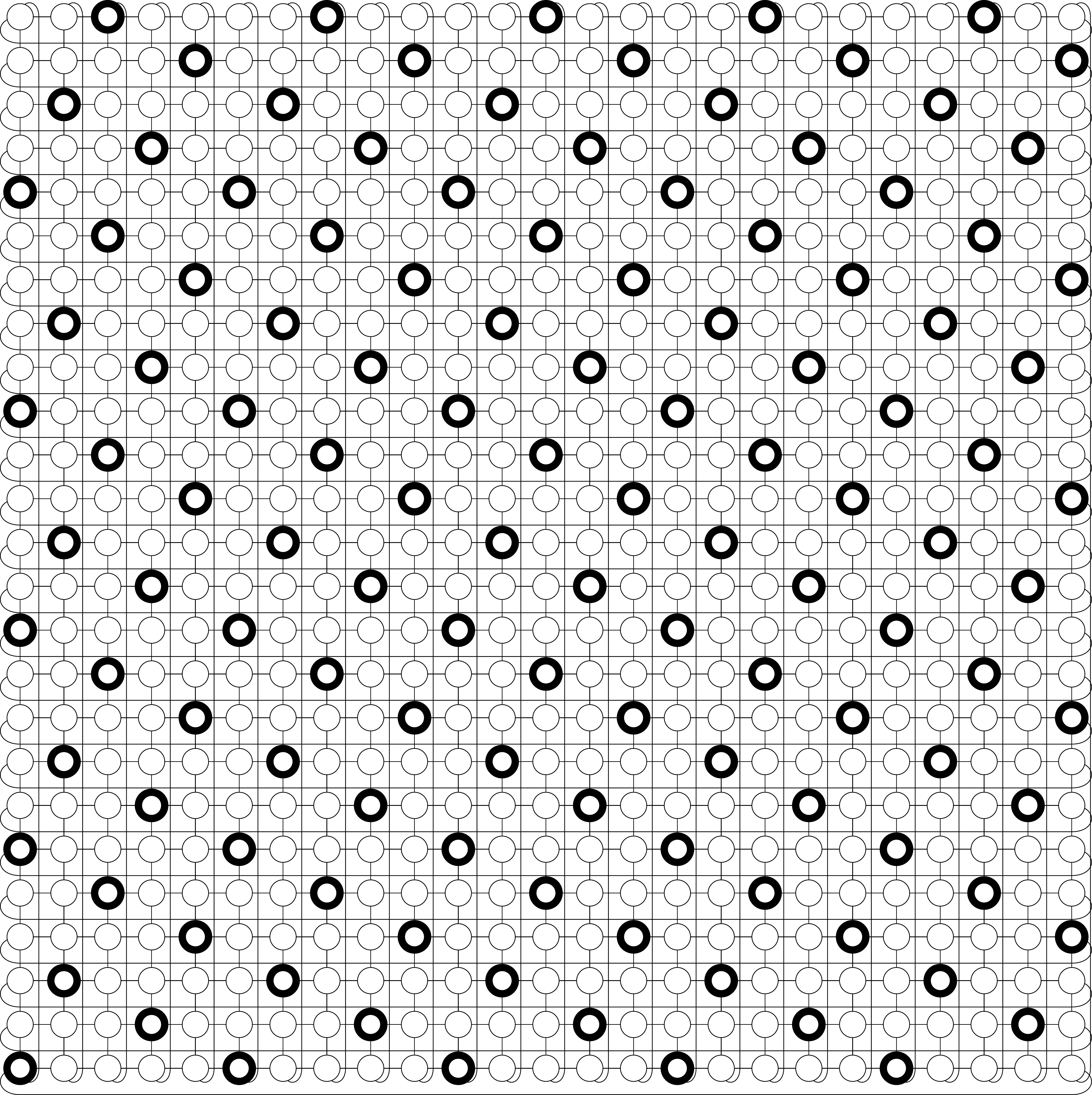}
	\end{center}
    \caption{
		1-Perfect and 2-identifying over the torus graph.
		}
    \label{fig:2identifT}
\end{figure}

\end{example}

\begin{example} Let us consider the Kronecker graph product of cycles of lengths 20 and 5. In \cite{Zerovnik}, it was shown that the code $\cC = \langle \vectwo{3}{1}, \vectwo{-1}{3} \rangle $ is a 1-perfect code over this graph. It can be seen that this graph is $\mathcal{G}(M; A)$, with matrix $M = \begin{pmatrix}20&0\\0&5\end{pmatrix}$ and set of adjacencies $A = \{ \vectwo{1}{-1}, \vectwo{-1}{1} \}$. As it was proved in \cite{ISIT10}, this graph is isomorphic to $\mathcal{G}(M')$, where $M' = \begin{pmatrix}10&-5\\10&5\end{pmatrix}$. Now, using the previous method it is obtained that the code $\cC = \langle \vectwo{1}{2}, \vectwo{-2}{1} \rangle $ is 1-perfect, but also 2-identifying in $\mathcal{G}(M)$. Note that the codes are the same by the graph isomorphism, but our method allows us to determine also the identifying radius of the code. In Figure \ref{fig:perfectKron}, the Kronecker product and the code are represented. In Figure \ref{fig:2identK} its isomorphic Cayley graph and the transformed code are shown.

\begin{figure}
\begin{center}
\begin{tikzpicture}[
	base/.style={circle,draw,inner sep=1ex},
	normal/.style={},
	code/.style={ultra thick},
	x=4ex,y=4ex,
	]
\pgfmathsetseed{1}
\foreach\x in {0,...,19}
\foreach\y in {0,...,4}
{
\expandafter\gdef\csname estilo\x_\y\endcsname{normal}
\expandafter\gdef\csname index\x_\y\endcsname{-1}
}

\foreach\i in {0,...,19}
{
\pgfmathparse{int(mod(3*\i,20))}\let\x=\pgfmathresult
\pgfmathparse{int(mod(\i,5))}\let\y=\pgfmathresult
\expandafter\xdef\csname estilo\x_\y\endcsname{code}
\expandafter\xdef\csname index\x_\y\endcsname{\i}
\draw[thick] (\x,\y) +(45:1.4) -- +(45:-1.4) +(135:1.4) -- +(135:-1.4);

\pgfmathparse{int(mod(20+\x+1,20)}\let\px=\pgfmathresult
\pgfmathparse{int(mod(5+\y+1,5)}\let\py=\pgfmathresult
\expandafter\xdef\csname index\px_\py\endcsname{\i}
\draw[thick] (\px,\py) -- +(45:-1.4);
\pgfmathparse{int(mod(20+\x-1,20)}\let\px=\pgfmathresult
\pgfmathparse{int(mod(5+\y+1,5)}\let\py=\pgfmathresult
\expandafter\xdef\csname index\px_\py\endcsname{\i}
\draw[thick] (\px,\py) -- +(135:-1.4);
\pgfmathparse{int(mod(20+\x+1,20)}\let\px=\pgfmathresult
\pgfmathparse{int(mod(5+\y-1,5)}\let\py=\pgfmathresult
\expandafter\xdef\csname index\px_\py\endcsname{\i}
\draw[thick] (\px,\py) -- +(135:1.4);
\pgfmathparse{int(mod(20+\x-1,20)}\let\px=\pgfmathresult
\pgfmathparse{int(mod(5+\y-1,5)}\let\py=\pgfmathresult
\expandafter\xdef\csname index\px_\py\endcsname{\i}
\draw[thick] (\px,\py) -- +(45:1.4);
}

\foreach\x in {0,...,19}
\foreach\y in {0,...,4}
{
\edef\index{\csname index\x_\y\endcsname}
\node[base,\csname estilo\x_\y\endcsname,fill=white] at (\x,\y) {};
}
\end{tikzpicture}
 \caption{1-perfect code over Kronecker product of cycles.}
 \label{fig:perfectKron}
\end{center}
\end{figure}
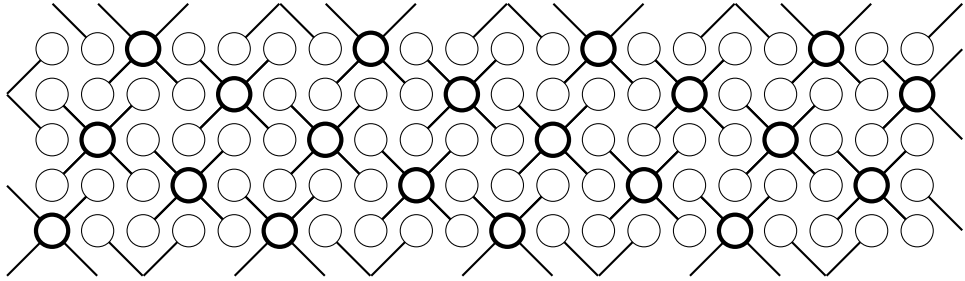

\begin{figure}
    \begin{center}
    \includegraphics[width=.6\columnwidth]{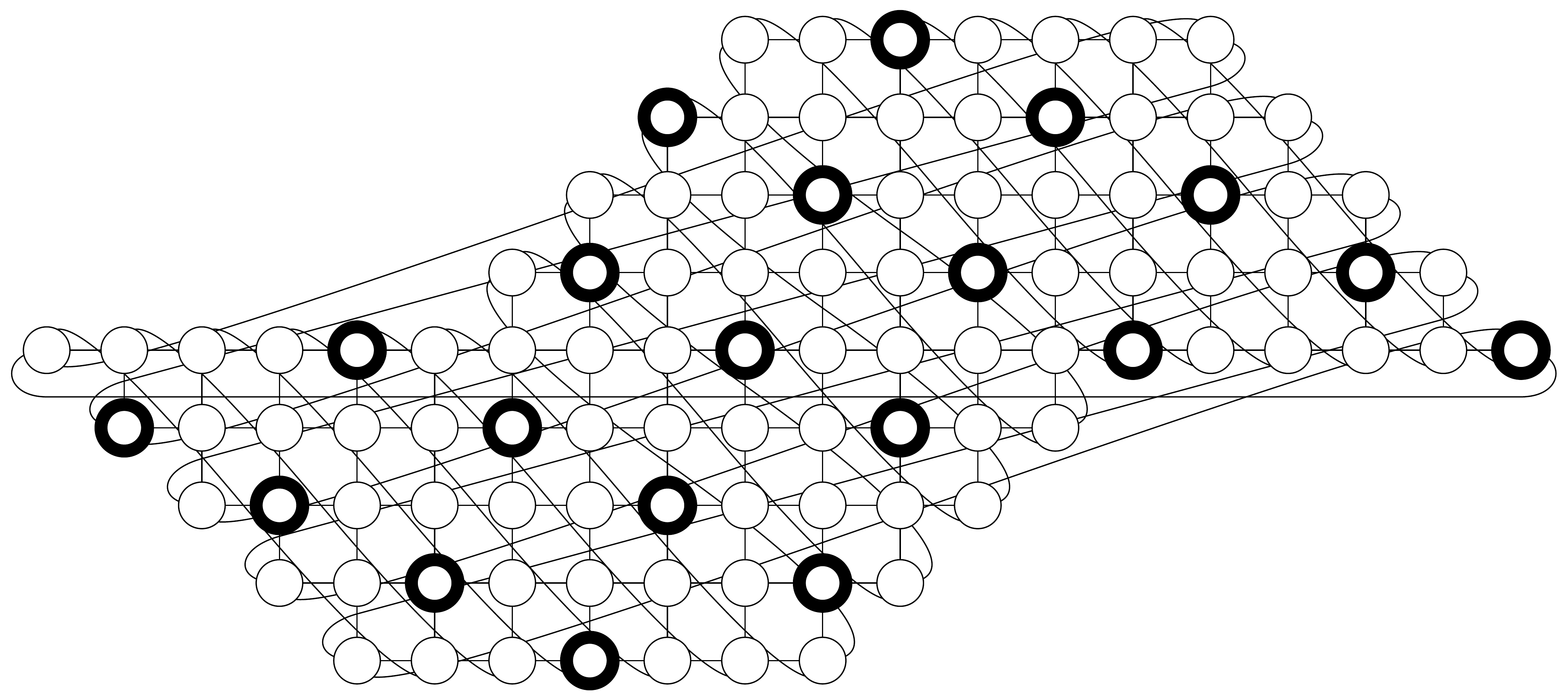}
	\end{center}
    \caption{
		2-identifying and 1-perfect code over $\mathcal{G}(M)$.
		}
    \label{fig:2identK}
\end{figure}

\end{example}

\subsection{Second Case:  $\gcd(2t, d)=2$}

Let $t$ and $d$ be two positive integers such that $\gcd(2t, d)=2$. Again, let us consider the code
 $\cC = \langle \vectwo{t}{t+d}, \vectwo{-(t+d)}{t}\rangle$ over $\mathbb{Z}^2$ and let us denote by $N=\det(\begin{pmatrix}t&-(t+d)\\t+d&t\end{pmatrix})=2t^2+2td+d^2$ the volume of its fundamental parallelogram. It can be proved that:

\begin{lemma}\label{lemma:codewords-par}$\vectwo{N/2}{N/2}$ and $\vectwo{N/2}{-N/2}$ are codewords in $\cC$.
\end{lemma}

\begin{proof} The proof is just to check that the following equalities hold:
$$\vectwo{N/2}{N/2}=\begin{pmatrix}t&-(t+d)\\t+d&t\end{pmatrix}\vectwo{t+d/2}{-d/2},$$
$$\vectwo{N/2}{-N/2}=\begin{pmatrix}t&-(t+d)\\t+d&t\end{pmatrix}\vectwo{-d/2}{t+d/2}.$$
\end{proof}

\begin{lemma}\label{lemma:2fijas} For all $v\in \mathbb{Z}^2$, $B_r(v)$ is $r$-fixed by $\cC$,
where $r=t^2+dt+\frac{d^2}{2}$ and $\gcd(2t, d)=2$.
\end{lemma}

\begin{proof} Analogously to the proof of Lemma \ref{lemma:fijas-impar} the following system of Diophantine equations is obtained:
$$
\left\{
\begin{array}{l}
x=at-b(t+d)-v_1\\
y=a(t+d)+bt-v_2
\end{array}\right.
$$

Therefore, $x+y=a(2t+d)-bd-v_1-v_2$. By hypothesis, $\gcd(2t+d,d)=\gcd(2t,d)=2$,
which implies that there exist $a_0,b_0 \in \mathbb{Z}$ such that
$a_0(2t+d)-b_0d=r+v_1+v_2$ or $a_0(2t+d)-b_0d=r+1+v_1+v_2$.
Thus, there also exist $x_0,y_0 \in \mathbb{Z}$ such that $x_0+y_0=r$ or $x_0+y_0=r+1$.
The set of solutions is $\{a'=a_0+\frac{d}{2}\lambda, b'=b_0+(t+\frac{d}{2})\lambda \ | \ \lambda \in \mathbb{Z} \}$.
Therefore, let $x'=a_0t-b_0(t+d)+t\frac{d}{2}\lambda-(t+\frac{d}{2})(t+d)\lambda-v_1
=x_0-\lambda \frac{N}{2}$ and in the same way let $y'=y_0+\lambda \frac{N}{2}$.\par

Now, let us set $\lambda$ such that $0< x' \leq \frac{N}{2}=r$.
Hence, $y'<x'+y'\leq r+y'$.
If $x'+y'=r$ then $0\leq y'<r$ and $c$ is at $\URi(v)$,
otherwise $x'+y'=r+1$ with $0<y'<r+1$ and $c$ is at $\URe(v)$.
\end{proof}

\begin{remark} \label{remark:strong-fixed} Note that it has been proved something stronger than being $r$-fixed, that is, every ball has a codeword in each of the four sets
$\ULi,\URi,\BLi,\BRi$ or in each of the four sets $\ULe,\URe,\BLe,\BRe$.
\end{remark}

\begin{theorem}\label{thm:identifying-par} Let $t, d$ be positive integers such that $\gcd(2t,d)=2$. Then, $\cC = \langle \vectwo{t}{t+d},\vectwo{-(t+d)}{t}\rangle$ is an $r$-identifying code in $\mathbb{Z}^2$ for $r=t^2+dt+\frac{d^2}{2}$.
\end{theorem}

\begin{proof} By previous Lemma \ref{lemma:2fijas}, all $r$-balls are $r$-fixed. Now, it has to be proved that any $v\in \mathbb{Z}^{2}$ is $r$-separated from $w=v+\vectwo{x}{y}$, for any $\vectwo{x}{y} \in \mathbb{Z}^{2}$.
First, as in the proof of Theorem \ref{thm:identifying} it can be assumed that $x>y\geq 0$.
As a consequence of Remark \ref{remark:strong-fixed} it can be assumed that there are codewords in each of $\ULe(v),\URe(v),\BLe(v),\BRe(v)$.
Let $c_1\in\cC\cap\BRe(v)$ which is not in $K_r(v)$. If $c_1$ is in $K_r(w)$ the proof is finished.

Since $c_1\in \BRe(v),$ this implies that it can be expressed as $c_1=v+\vectwo{a}{a-r-1}$. Then, $$c_1+\vectwo{-(t+d)}{t}=v+\vectwo{a-(t+d)}{a+t-r-1}$$
	is a codeword which is further from $K_r(w)$ than $c_1$. When $a>\frac{d}{2}$, it is in $K_r(v)$.

Proceeding analogously, but for $\URe(v)$, it is obtained that $c_2=v+\vectwo{b}{r-b+1}$.
Then, let us consider
	$$c_2+\vectwo{-(t+d)}{t}-\vectwo{t}{t+d}=c_2+\vectwo{-2t-d}{-d}
		=v+\vectwo{b-2t-d}{r-b-d+1}
		$$
which is further from $K_r(w)$ than $c_2$ and when $b>t$ it is in $K_r(v)$.\par

Finally, the case when both codewords $c_1,c_2$ are in the left ranges has to be analyzed.
For that, when $1\leq a\leq\frac{d}{2}$ and $1\leq b\leq t$, with $c_1,c_2\in\cC$ is obtained that $c=c_1-c_2\in\cC$. Then, $c=c_1-c_2=\vectwo{a-b}{a+b-1-N}=\vectwo{x'}{y'}$ is in the rectangle defined by:
$$
	\{1-t\leq x'\leq \frac{d}{2}-1 < d,\ -N\leq y'\leq t+\frac{d}{2}-N-2 < t+d-N\}
$$
which includes the point $\vectwo{0}{-N}$. This rectangle is a subset of  the set of representatives of $\Z^2/\cC \Z^2$
\begin{multline*}
	\{1-2t\leq x'\leq -t, -N\leq y'\leq -N+t-1\} \cup\\
	\{ 1-t\leq x'\leq d, -N\leq y'\leq -N+t+d-1\}
\end{multline*}
and consequently, all its elements are different. Therefore, $c$ is the only element congruent with $\vectwo{0}{0}$ (using Lemma \ref{lemma:codewords-par}) and thus
$x'=0$ and $y'=-N$. Follows that $a=b=1$ and $c_1=\vectwo{1}{-r}$. If $x=1$ then $c_1$ separates $v$ from $w$.
If $x>1$ then $c_1+\vectwo{-\frac{N}{2}}{\frac{N}{2}}=\vectwo{1-r}{0}$
(which is another codeword by Lemma \ref{lemma:codewords-par}) separates $v$ from $w$.
\end{proof}

\begin{remark}
Note that, if $N = \det(C)$ denotes the volume of the fundamental parallelogram of the lattice code $\cC$, in both cases $r = \lfloor \frac{N}{2} \rfloor$ is the identifying radius.
In the first case, that is $\gcd(2t, d)=1$, by computational analysis it has been observed that the obtained radius in Theorem \ref{thm:identifying} seems to be the minimal one. However, in this second case, smaller radii have been found that make the code also identifying. For example, it has been observed that in many cases $t^2+(d-1)t+\frac{d^2-d+1}{2}$ is also a possible radius of identification, which is smaller than the one given in Theorem \ref{thm:identifying-par}.
\end{remark}

\subsection{Third Case:  $\gcd(2t, d) \geq 3$}

In previous subsections it has been computed the radius $r$ such that the code $\cC = \langle \vectwo{t}{t+d}, \vectwo{-(t+d)}{t}\rangle$ is $r$-identifying over the lattice $\mathbb{Z}^{2}$, wherever $t$ and $d$ are two positive integers such that $\gcd(2t, d) \leq 2$. In this subsection it is established that, in any other election for $t$ and $d$, the lattice code $\cC$ is not $r$-identifying, for any $r \in \mathbb{Z}$.

\begin{lemma}\label{lemma:diagonals} For any group code $\cC$ over $\Z^2$ for which there exists $\vectwo{x}{y}$ such that there is no codeword in any of the subsets:
\begin{align*}
\ell_1&= \left \{\vectwo{x+k}{y+k} \ |\ k\in\Z \right \},\\
\ell_2&= \left \{\vectwo{x+k}{y-k} \ |\ k\in\Z \right \},\\
\ell_3&= \left \{\vectwo{x+k}{y+k+1} \ |\ k\in\Z \right \}\text{ and}\\
\ell_4&=\left \{\vectwo{x+k}{y-k+1} \ |\ k\in\Z \right \},
\end{align*}
then, $\cC$ is not an identifying code for any radius.
\end{lemma}

\begin{proof} The proof is made by \textsl{reductio ad absurdum}. Therefore, let us assume that $\cC$ identifies for the radius $r \in \N$.
Let $v=\vectwo{x}{y} \in \mathbb{Z}^2$ be the vertex of the hypothesis.
The translations of $\ell_i$ are referred as diagonals and two diagonals $\ell_a,\ell_b$ are contiguous when $\ell_a=\ell_b\pm \vectwo{0}{1}$. For example,
$\ell_1$ and $\ell_3$ are two contiguous diagonals. First note that for any two contiguous diagonals, at least one must be codewordless. Otherwise, since $\cC$ is a group, there would be codewords in all diagonals. Moreover, if there is a codeword in a diagonal $\ell$, then the two contiguous diagonals in both directions
($\ell+\vectwo{0}{1},\ell+\vectwo{0}{2},\ell-\vectwo{0}{1},\ell-\vectwo{0}{2}$) must be codewordless.\par

Now, consider the word $w=\vectwo{x}{y-r-1}$. The following diagonals at the bottom
of the ball $B_r(w)$ are denoted by:
\begin{align*}
\ell_1'&= \left \{\vectwo{x+k}{y+k-2r-2} \ | \ k\in\Z \right \}&
\ell_2'&= \left \{\vectwo{x+k}{y-k-2r-2} \ | \ k\in\Z\right\}\\
\ell_3'&= \left \{\vectwo{x+k}{y+k-2r-1} \ | \ k\in\Z \right\}&
\ell_4'&= \left \{\vectwo{x+k}{y-k-2r-1} \ | \ k\in\Z \right \}\\
\ell_5'&= \left \{\vectwo{x+k}{y+k-2r} \ | \ k\in\Z \right \}&
\ell_6'&= \left \{\vectwo{x+k}{y-k-2r} \ | \ k\in\Z \right \}
\end{align*}

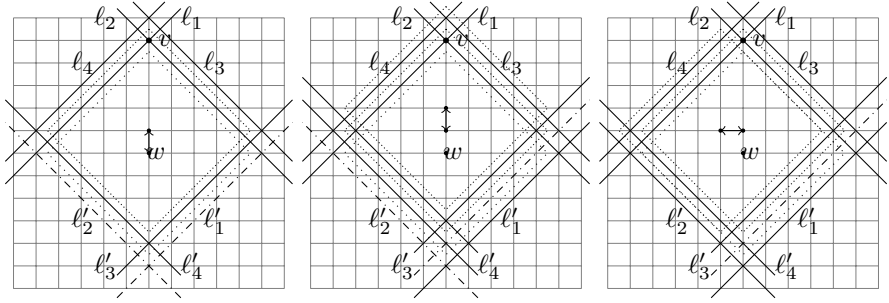
\begin{figure}
	\begin{center}
		\begin{tikzpicture}
			[scale=.3]
		\def\r{4}
		\draw[help lines] (-6,-6) grid (6,6);
		\filldraw (0,\r+1) circle (3pt) node[anchor=west] {$v$};
		\filldraw (0,1) circle (2pt);
		\filldraw (0,0) circle (2pt) ++(2ex,0) node {$w$};
		\draw[<->] (0,1)-- node[auto]{} (0,0);
		\draw[dotted] (0,\r+.5) -- (\r+.5,0) -- (0,-\r-.5) -- (-\r-.5,0) --cycle;
		\draw[densely dotted] (0,\r+1.5) -- (\r+.5,1) -- (0,-\r+.5) -- (-\r-.5,1) --cycle;
		\draw (0,\r+1) +(45:2) -- +(45:-9);
		\draw (0,\r+2) +(45:1) -- +(45:-9);
		\draw (0,\r+1) +(135:2) -- +(135:-9);
		\draw (0,\r+2) +(135:1) -- +(135:-9);
		\node[anchor=west] at (1,\r+2) {$\ell_1$};
		\node[anchor=east] at (-1,\r+2) {$\ell_2$};
		\node[anchor=west] at (2,\r) {$\ell_3$};
		\node[anchor=east] at (-2,\r) {$\ell_4$};

		\draw (0,-\r) +(45:-2) -- +(45:9);
		\draw[dashdotted] (0,-\r-1) +(45:-2) -- +(45:9);
		\draw (0,-\r) +(135:-2) -- +(135:9);
		\draw[dashdotted] (0,-\r-1) +(135:-2) -- +(135:9);
		\node[anchor=east] at (-1,-\r-1) {$\ell_3'$};
		\node[anchor=west] at (1,-\r-1) {$\ell_4'$};
		\node[anchor=west] at (2,-\r+1) {$\ell_1'$};
		\node[anchor=east] at (-2,-\r+1) {$\ell_2'$};
		\end{tikzpicture}
		\begin{tikzpicture}
			[scale=.3]
		\def\r{4}
		\draw[help lines] (-6,-6) grid (6,6);
		\filldraw (0,\r+1) circle (3pt) node[black,anchor=west] {$v$};
		\filldraw (0,2) circle (2pt);
		\filldraw (0,1) circle (2pt);
		\filldraw (0,0) circle (2pt) ++(2ex,0) node {$w$};
		\draw[<->] (0,2)-- node[auto]{} (0,1);
		\draw[dotted] (0,\r+.5) -- (\r+.5,0) -- (0,-\r-.5) -- (-\r-.5,0) --cycle;
		\draw[densely dotted] (0,\r+1.5) -- (\r+.5,1) -- (0,-\r+.5) -- (-\r-.5,1) --cycle;
		\draw[densely dotted] (0,\r+2.5) -- (\r+.5,2) -- (0,-\r+1.5) -- (-\r-.5,2) --cycle;
		\draw (0,\r+1) +(45:2) -- +(45:-9);
		\draw (0,\r+2) +(45:1) -- +(45:-9);
		\draw (0,\r+1) +(135:2) -- +(135:-9);
		\draw (0,\r+2) +(135:1) -- +(135:-9);
		\node[anchor=west] at (1,\r+2) {$\ell_1$};
		\node[anchor=east] at (-1,\r+2) {$\ell_2$};
		\node[anchor=west] at (2,\r) {$\ell_3$};
		\node[anchor=east] at (-2,\r) {$\ell_4$};

		\draw[dashdotted] (0,-\r) +(45:-2) -- +(45:9);
		\draw (0,-\r-1) +(45:-2) -- +(45:9);
		\draw (0,-\r) +(135:-2) -- +(135:9);
		\draw[dashdotted] (0,-\r-1) +(135:-2) -- +(135:9);
		\node[anchor=east] at (-1,-\r-1) {$\ell_3'$};
		\node[anchor=west] at (1,-\r-1) {$\ell_4'$};
		\node[anchor=west] at (2,-\r+1) {$\ell_1'$};
		\node[anchor=east] at (-2,-\r+1) {$\ell_2'$};

		\draw[] (0,-\r+1) +(45:-2) -- +(45:9);
		\draw[] (0,-\r+1) +(135:-2) -- +(135:9);
		\end{tikzpicture}
		\begin{tikzpicture}
			[scale=.3]
		\def\r{4}
		\draw[help lines] (-6,-6) grid (6,6);
		\filldraw (0,\r+1) circle (3pt) node[black,anchor=west] {$v$};
		\filldraw (0,1) circle (2pt);
		\filldraw (-1,1) circle (2pt);
		\filldraw (0,0) circle (2pt) ++(2ex,0) node {$w$};
		\draw[<->] (0,1)-- node[auto]{} (-1,1);
		\draw[dotted] (0,\r+.5) -- (\r+.5,0) -- (0,-\r-.5) -- (-\r-.5,0) --cycle;
		\draw[densely dotted] (0,\r+1.5) -- (\r+.5,1) -- (0,-\r+.5) -- (-\r-.5,1) --cycle;
		\draw[densely dotted] (-1,\r+1.5) -- (\r-.5,1) -- (-1,-\r+.5) -- (-\r-1.5,1) --cycle;
		\draw (0,\r+1) +(45:2) -- +(45:-9);
		\draw (0,\r+2) +(45:1) -- +(45:-9);
		\draw (0,\r+1) +(135:2) -- +(135:-9);
		\draw (0,\r+2) +(135:1) -- +(135:-9);
		\node[anchor=west] at (1,\r+2) {$\ell_1$};
		\node[anchor=east] at (-1,\r+2) {$\ell_2$};
		\node[anchor=west] at (2,\r) {$\ell_3$};
		\node[anchor=east] at (-2,\r) {$\ell_4$};

		\draw[dashdotted] (0,-\r) +(45:-2) -- +(45:9);
		\draw (0,-\r-1) +(45:-2) -- +(45:9);
		\draw (0,-\r) +(135:-2) -- +(135:9);
		\draw (0,-\r-1) +(135:-2) -- +(135:9);
		\node[anchor=east] at (-1,-\r-1) {$\ell_3'$};
		\node[anchor=west] at (1,-\r-1) {$\ell_4'$};
		\node[anchor=west] at (2,-\r+1) {$\ell_1'$};
		\node[anchor=east] at (-2,-\r+1) {$\ell_2'$};

		\draw (0,-\r+1) +(45:-2) -- +(45:9);
		\end{tikzpicture}
	\end{center}
	\caption{Possibilities for the diagonals}
	\label{fig:diagonals}
\end{figure}

Note that $w$ is at distance $r+1$ from $v$, and hence from $\ell_1$ and $\ell_2$. Also, $v$ is at distance $r+2$ from $\ell_3$ and $\ell_4$. Since $\ell_1'$ and $\ell_2'$ are at distance $r+1$ from $w$, it will be said that they are \emph{outer} of $B_r (w)$. Also, since $\ell_3',\ell_4'$ at distance $r$ it will be said that they are \emph{inner} of $B_r (w)$. At least one of $\ell_1',\ell_3'$ is codewordless and the same happens for $\ell_2',\ell_4'$. In Figure \ref{fig:diagonals} an example of the relative locations of $v$, $w$ and the diagonals is shown. Now, three different cases are considered according to the location of the codewords with respect to the diagonals, that is $\ell_1',\ell_2'$ outside and $\ell_3',\ell_4'$ inside:
\begin{enumerate}
\item All the diagonals with codewords (0, 1 or 2) are outer diagonals. In the figure $\ell_1'$ and $\ell_2'$ are dashed, indicating
	that they have some codeword, all other diagonals are solid, indicating that they have no codewords.
	As diagonals $\ell_3'$ and $\ell_4'$ are codewordless it is obtained that $K_r(w)\subseteq K_r(w+\vectwo{0}{1})$.
	Since the same is true for $\ell_3$ and $\ell_4$, it is obtained that $K_r(w+\vectwo{0}{1})\subseteq K_r(w)$.
	Thus $K_r(w)=K_r(w+\vectwo{0}{1})$ and they are not separable. In the figure it can be seen that in the symmetric difference of the two balls there
	are only solid lines, which indicates the non separability.
\item Two diagonals with codewords are inner, or 1 is inner and 1 is outer. Since two contiguous diagonals above any diagonal with a codeword must be
	codewordless and one of $\ell_1',\ell_3'$ has a codeword, then $\ell_5'$ is codewordless.
	Respectively, since $\ell_2'$ or $\ell_4'$ has a codeword then $\ell_6'$ is codewordless.
	Now, we have the same position as in the previous item of the proof, but translated by $\vectwo{0}{1}$ and
	that $K_r(w+\vectwo{0}{1})=K_r(w+\vectwo{0}{2})$. Hence, they are not separable.
\item Only one diagonal with codewords is inner. Without loss of generality let us assume that $\ell_3'$ is the inner diagonal with some codeword.
	Therefore, it is obtained that $\ell_1',\ell_2',\ell_4'$ and $\ell_5'$ are codewordless. Analogously to the previous cases, but for
	$w+\vectwo{0}{1}$ and $w+\vectwo{-1}{1}$ implies that $K_r(w+\vectwo{0}{1})=K_r(w+\vectwo{-1}{1})$ are not separable.
\end{enumerate}
In all cases non-separable codewords have been found, which contradicts the $r$-identifiability of $\cC$
and thus the proof is completed.
\end{proof}

\begin{theorem} Let $t, d$ be two positive integers such that $\gcd(2t,d)\geq 3$. Then, the code
$\cC = \langle \vectwo{t}{t+d}, \vectwo{-(t+d)}{t}\rangle$ is not identifying for any radius over $\mathbb{Z}^{2}$.
\end{theorem}

\begin{proof}
A codeword $\vectwo{x}{y}$ of $\cC$ satisfies that there exist $a, b \in \mathbb{Z}$ such that:
$$
\left\{
\begin{array}{l}
x=at-b(t+d)\\
y=a(t+d)+bt
\end{array}\right.
$$

Let $h$ be an integer such that $y=x+h$, so $h=ad+b(2t+d)$ is obtained.
Hence, $\gcd(2t+d,d)=\gcd(2t,d)$ divides $h$.
This implies that the contiguous diagonals $\{\vectwo{k}{k+h} \ | \ k\in\Z\}$ are codewordless for $0<h<\gcd(2t,d)$.\par

When $\gcd(2t,d)\geq 3$, at least two contiguous diagonals have not codeword.
Since the code is invariant under rotations, contiguous orthogonal diagonals are also codewordless.
Thus $\cC$ satisfies the hypothesis of the Lemma \ref{lemma:diagonals} and therefore, it is not
identifying for any radius.
\end{proof}

\subsection{On the Density of the Construction} \label{subsec:Densidad}

Several papers have been developed for the computation of identifying codes attaining the minimal density bounds. To define the density of a code over $\Z^2$, let us denote by $Q_n$ the set of vertices $\vectwo{x}{y} \in \mathbb{Z}^{2}$ with $|x| \leq n$ and $|y| \leq n$. Then, the \emph{density} of a code $\cC$ is defined as:

$$D(\cC) = \lim \sup _{n \rightarrow \infty} \frac{\cC \cap Q_n}{Q_n}.$$

In \cite{Charonsmallradius} there is a compilation of lower and upper bounds of the best possible density for codes in some infinite regular graphs,
in terms of the identifying radius. For the case of the square grid it is obtained that:

	$$\frac{3}{8r+4}\leq D(r)\leq\begin{cases}\frac{2}{5r}&\text{if $r$ even,}\\
	\frac{2r}{5r^2-2r+1}&\text{if $r$ odd.}\end{cases}$$

In particular, in \cite{Benhaimexact} it was proved for 1-identifying codes that:

\begin{theorem} \cite{Benhaimexact} The minimum density of 1-identifying codes in $\mathbb{Z}^{2}$ is $\frac{7}{20}.$
\end{theorem}

In the following, the density for our construction of identifying codes is calculated. Although the values are not far from the optimal one, they are established here since they can be exactly computed in order to complete our study.

\begin{lemma} Let $\mathcal{G}(M)$ be such that $\cC = \langle \vectwo{t}{t+d},
\vectwo{-(t+d)}{t}\rangle$ is an $r$-identifying code.  Then, the density of $\cC$ is

$$D(r) = \begin{cases}\frac{1}{2r+1}&\text{if $r$ is odd,}\\
	\frac{1}{2r}&\text{if $r$ is even.}\end{cases}$$

\end{lemma}

\begin{proof} Let $C = \begin{pmatrix}t&-(t+d)\\t+d&t\\\end{pmatrix}.$ By definition of density,:

$$D(r) = \lim_{n \rightarrow \infty }\frac{|\mathcal{C} \cap Q_{n}|}{|Q_{n}|} = \frac{|\mathcal{C}|}{|V|} = \frac{\frac{\det(M)}{\det(C)}}{\det(M)} = \frac{1}{\det(C)}.$$

Now, since $\det(C) = t^2+(t+d)^2 = 2r+1$ or $2r$, the result is straightforwardly obtained.
\end{proof}

\begin{remark} The corresponding code in $\Z^2$ has the same density.
\end{remark}

Although the constructed codes have a little higher density than the best constructions
they form a wide family and its regularity may have practical applications.
This fact will be highlighted in next section, where it is taken advantage
of the inclusion of perfect codes in our family of identifying codes.

\section{Adaptive Identification in Cayley Graphs} \label{sec:algoritmo}

In this section the implications of the constructed codes in adaptive identification is considered. Firstly, the approach for adaptive identification presented in \cite{Benhaimadaptive} and \cite{Benhaimadaptive2} is summarized. Then, our proposal for identifying faults with the codes presented in this paper is stated. Finally, a detailed example of this proposal is addressed.\par

In \cite{Benhaimadaptive, Benhaimadaptive2}, \emph{adaptive identification} is defined as a game. In that game, a player secretly chooses a vertex (or none) of the graph and another player tries to find it by making successive queries. For any vertex $v$ of the graph, the queries are of the type:
\begin{center}
	``Is the chosen vertex in $B_r(v)$?''
\end{center}
and these queries are answered truthfully.\par

If $\mathcal C$ is a $r$-identifying code, then the chosen vertex can be found by making queries only to vertices in $\mathcal C$. As proposed in \cite{Benhaimadaptive}, the identification can be done in two stages,
in order to minimize the number of performed queries. The first stage consist on making queries in a $r$-covering code that is intended to be near perfect. Thus, this $r$-covering
is used to find out the ball $B_r(c)$ which contains the chosen vertex. In the second stage, using the $r$-identifying code, the vertex is determined. That is, two different codes with
the same parameter $r$ being the covering and identification radii, respectively, are used to identify the selected vertex. As a consequence, at least three kind of vertices coexist in the graph: vertices belonging to the $r$-covering, vertices belonging to the $r$-identifying and the rest of the vertices of the graph.\par

The proposal in the present paper is to use the same code for both stages, instead of having two different codes. As it was proved in Corollary \ref{col:perfecto}, $t$-perfect codes of the form $\cC = \langle \vectwo{t}{t+1}, \vectwo{-(t+1)}{t}\rangle $ are also $r$-identifying codes for $r=t(t+1)$. Moreover, although they are not the most dense as identifying codes, they are not very far from the bound. Thus, in this alternative configuration there will be only two kind of vertices: codewords and non-codewords. On the contrary, there will be two kind of queries (for different radius $t$ and $r$), both for any codeword $c\in\mathcal C$. The first query is of the type:
\begin{center}
	``Is the chosen vertex in $B_t(c)$?''
\end{center}
and the second one of the type:
\begin{center}
	``Is the chosen vertex in $B_r(c)$?''
\end{center}
Therefore, in the first stage of the method the perfect code will be used to determine if there is a faulty vertex. If so, the vertex that detects the faulty vertex initiates the identification process.\par

Identifying codes have been motivated several times as a good strategy to locate faulty nodes in a multiprocessor system, although they have not been applied at the moment. In such a system, the existence of a faulty node is the exceptional case. Hence, it is important that the first stage of the identification, which is just the detection of the fault, is optimized. This is the reason why in our proposal the detection of the faulty node and consequently  its bound inside a ball or radius $r$ is made with a perfect code. Then, since the second stage (to locate the faulty vertex exactly) would be the exceptional situation, the requirements on the identification radius have been relaxed.\par

Finally, to illustrate this method, a example of this application is detailed next. Let us consider the graph $\mathcal{G}(M)$ defined by $M=\bigmattwo{28}{-4}{4}{28}$, which is graphically represented in Figure \ref{fig:perfectidentifying}. Let us also consider the code defined by the group  $\cC=\langle\vectwo{3}{4},\vectwo{-4}{3}\rangle$ which is also represented in the figure. By Corollary \ref{col:perfecto}, this is a 3-perfect code and a 12-identifying code over $\mathcal{G}(M)$. Note that the graph has $\det(M) = 800$ vertices and the code has $\frac{800}{25}=32$ codewords. Then, following the adaptive strategy detailed above, for each codeword $c \in \cC$ it is periodically checked if there is a faulty vertex in $B_3(c)$ by performing queries of the first type. Then, if and only if the obtained answer is positive, the second stage of the identification is started by the same codeword.\par

Assume that for example node $\vectwo{2}{1}$ has failed, which is represented in bold. Then, the first stage of the procedure detects a faulty vertex and bound it in a ball of 25 vertices, whose center is in this case the codeword $\vectwo{0}{0}$, as can be observed in Figure \ref{fig:usingperfect}. Then, that codeword starts the identification stage by making queries in $B_{12}(c)$ balls, as shown in Figure \ref{fig:usingidentifiying}. Therefore, with a binary decision tree of depth $\lceil\log_2(25)\rceil=5$ as the one in Figure \ref{fig:tree}, the faulty node can be isolated. Each vertex of the tree represents the realization of a query of the second type, continuing by the branch which matches the answer. Branches to the left represent positive answers to the question ``Is the faulty vertex in $B_{12}(c)$''. On the contrary, branches to the right represent negative answers. In this particular case, node $\vectwo{2}{1}$ is identified by the path of queries:

$$ \vectwo{2}{11} \rightarrow \vectwo{-11}{2} \rightarrow \vectwo{8}{-6} \rightarrow \vectwo{6}{8} \rightarrow \vectwo{10}{5}.$$

In non-adaptive identification all the codewords are checked, that is $|\mathcal C|=32$ queries of radius $r$ are done. Our proposal makes 32 queries of radius $t=3$ and then, when necessary, 5 queries of radius $r=12$.

\begin{figure}
    \begin{center}
    \includegraphics[width=.6\columnwidth]{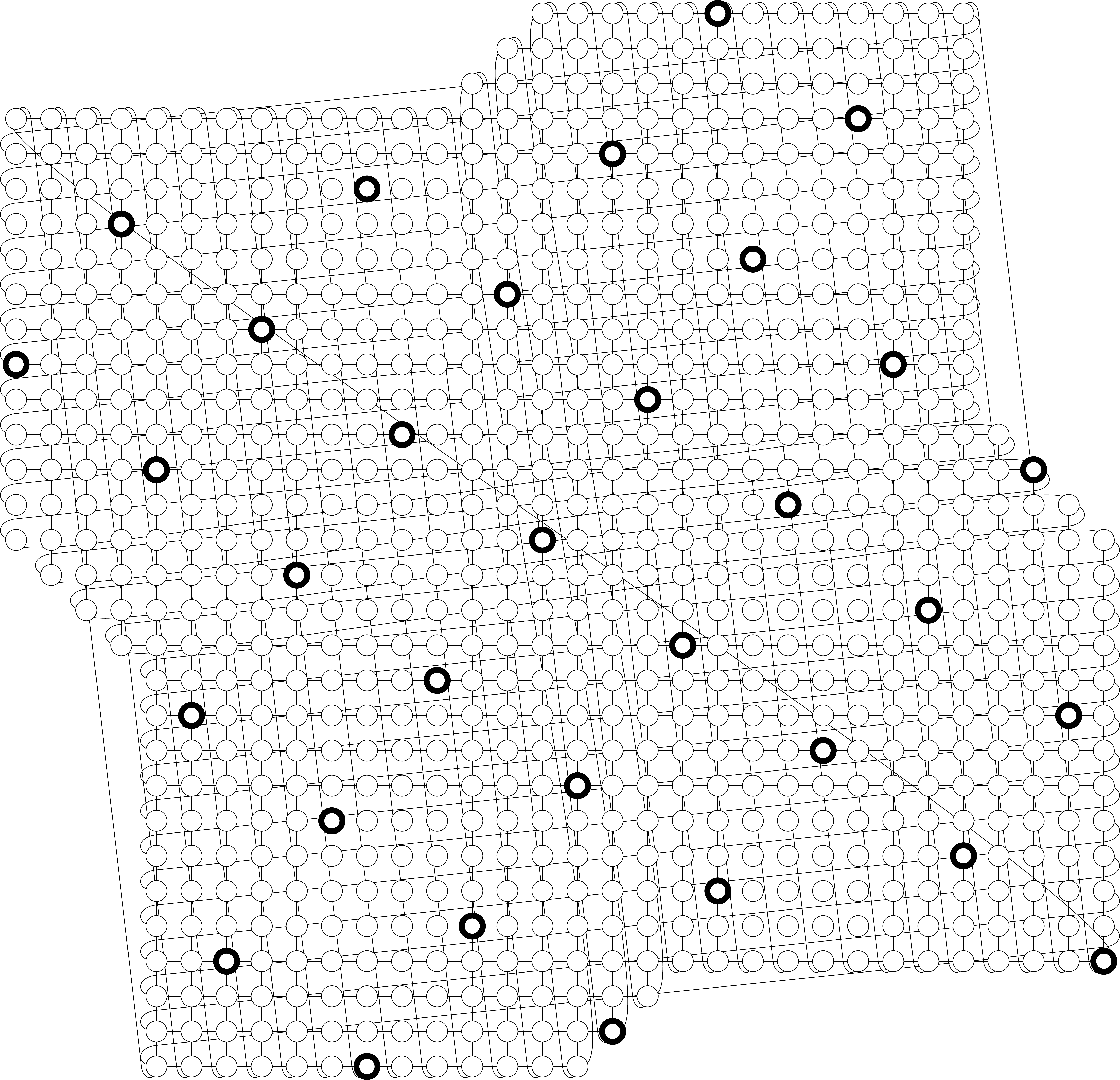}
	\end{center}
    \caption{A 3-perfect and 12-identifying code.}
    \label{fig:perfectidentifying}
\end{figure}

\begin{figure}
    \begin{center}
	\begin{tikzpicture}
	\pgftext[base,left]{\includegraphics[width=.6\columnwidth]{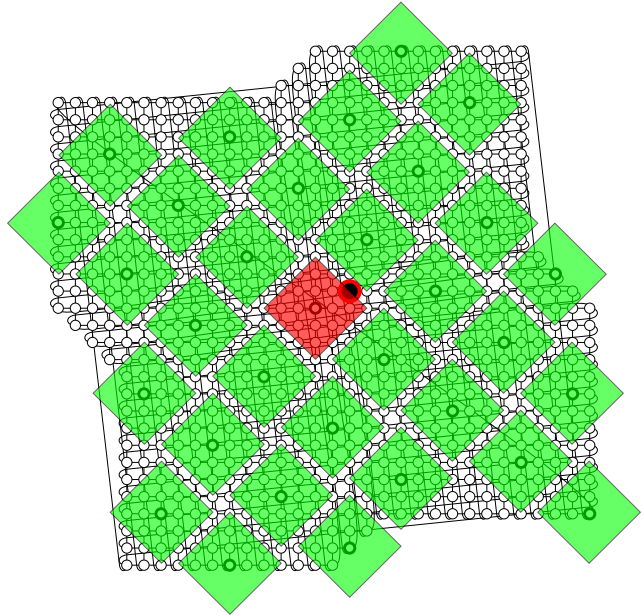}}
    \path
        (current bounding box.south west) coordinate (SW)
        (current bounding box.south east) coordinate (SE)
        (current bounding box.north west) coordinate (NW)
        (current bounding box.north east) coordinate (NE)
		($(SW)!.5!(NE)$) coordinate (M)
		($(NE)-(SW)$) coordinate (DIAG)
		let \p1=(DIAG) in
		(\x1,0) coordinate (HOR) 
		(0,\y1) coordinate (VER) 
		(\x1/32,0) coordinate (steph)
		(0,\y1/31) coordinate (stepv)
		($(M)-.5*(steph)$) coordinate (C);
        ;
	\def\ball#1#2#3{\path[#3] let \p1=#1 in (\p1) ++($#2*(steph)$)
		-- ($(\p1)+#2*(stepv)$)
		-- ($(\p1)-#2*(steph)$)
		-- ($(\p1)-#2*(stepv)$)
		--cycle;}
	\def\word#1#2{($(C)+#1*(steph)+#2*(stepv)$)}
	\fill[draw=red,line width=1.5pt,fill=black] \word{2}{1} circle (3.8pt);
	{
		\foreach \i/\j/\c in {
			0/0/red,3/4/green,6/8/green,9/12/green,-3/-4/green,-6/-8/green,-9/-12/green,
			4/-3/green,7/1/green,10/5/green,1/-7/green,-2/-11/green,-5/-15/green,
			-4/3/green,-7/-1/green,-10/-5/green,-1/7/green,2/11/green,5/15/green,
			8/-6/green,11/-2/green,14/2/green,5/-10/green,2/-14/green,
			-8/6/green,-11/2/green,-5/10/green,
			-12/9/green,-15/5/green,
			12/-9/green,15/-5/green,16/-12/green}
		{
			\ball{\word{\i}{\j}}{3}{draw,fill=\c,opacity=.6,very thin}
		}
	}
	\ball{\word{0}{0}}{3}{draw,fill=red,opacity=.1,very thin}
	\end{tikzpicture}
	\end{center}
    \caption{Querying $B_t(c)$ for each codeword, locating a faulty vertex.}
    \label{fig:usingperfect}
\end{figure}

\begin{figure}
    \begin{center}
    \begin{tikzpicture}
    \pgftext[base,left]{\includegraphics[width=.6\columnwidth]{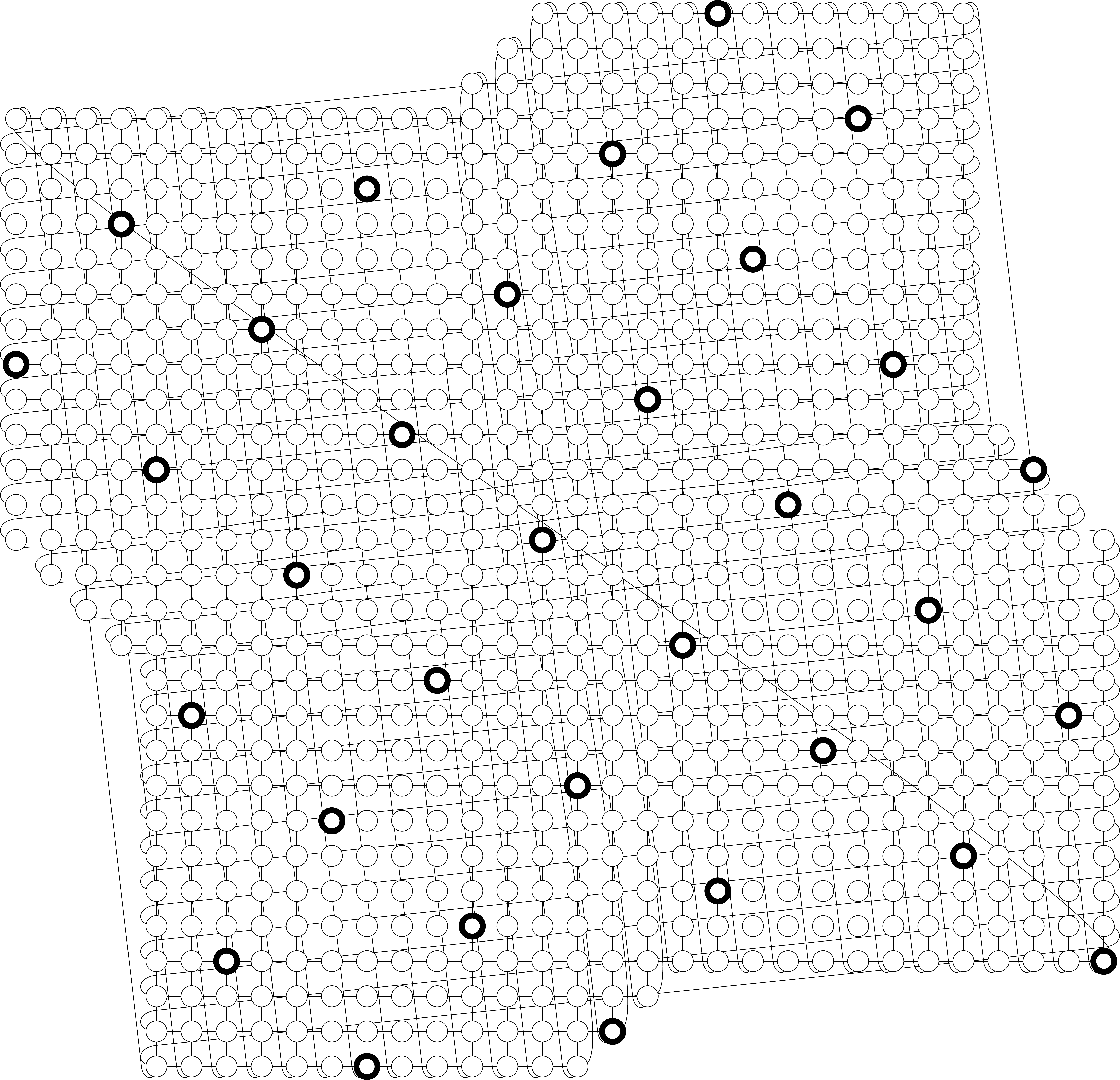}}
    \path
        (current bounding box.south west) coordinate (SW)
        (current bounding box.south east) coordinate (SE)
        (current bounding box.north west) coordinate (NW)
        (current bounding box.north east) coordinate (NE)
        ($(SW)!.5!(NE)$) coordinate (M)
        ($(NE)-(SW)$) coordinate (DIAG)
        let \p1=(DIAG) in
        (\x1,0) coordinate (HOR) 
        (0,\y1) coordinate (VER) 
        (\x1/32,0) coordinate (steph)
        (0,\y1/31) coordinate (stepv)
        ($(M)-.5*(steph)$) coordinate (C);
        ;
    \def\ball#1#2#3{\path[#3] let \p1=#1 in (\p1) ++($#2*(steph)$)
        -- ($(\p1)+#2*(stepv)$)
        -- ($(\p1)-#2*(steph)$)
        -- ($(\p1)-#2*(stepv)$)
        --cycle;}
    \def\word#1#2{($(C)+#1*(steph)+#2*(stepv)$)}
    \fill[draw=red,line width=1.5pt,fill=black] \word{2}{1} circle (3.8pt);
    \ball{\word{0}{0}}{3}{draw,fill=red,opacity=.1,very thin}
    \filldraw[draw=black,fill=yellow] \word{2}{11} circle (3.8pt);
    \ball{\word{2}{11}}{12}{draw,fill=red,opacity=.6,very thin}
    \end{tikzpicture}
    \end{center}
    \caption{Querying a ball $B_r(c)$ centered in a codeword.}
    \label{fig:usingidentifiying}
\end{figure}

\begin{figure}
{
	\begin{center}
\begin{tikzpicture}[
		scale=0.8,text centered,
		query/.style={thick,draw,circle,font=\scriptsize,inner sep=1pt,minimum width=6ex},
		response/.style={thin,draw,rectangle,font=\tiny,inner sep=0,minimum width=3ex,minimum height=1em},
		every child node/.style={query,draw=black,fill=none},
		edge from parent/.append style={thin,font=\bf\tiny},
		level distance=2.2cm,
		level 1/.style={sibling distance=11cm},
		level 2/.style={sibling distance=5cm},
		level 3/.style={sibling distance=2.5cm},
		level 4/.style={sibling distance=1.3cm},
		level 5/.style={sibling distance=0.6cm},
		faulty/.style={draw=black},
		fault-free/.style={draw=black!75},
		camino node/.style={},
		camino link/.style={},
		]
	\path[blue] (-7.5,-8) -- (0,.5) -- (7.5,-8);
	\begin{scope}[overlay]
	\node[query,camino node]{2,11} [grow=down,edge from parent fork down,transform canvas={scale=0.7}]
		child {node[camino node]{-11,2}
			child {node {6,8}
				child {node {-8,6}
					child {node {10,5}
						child {node[response] {0,3}
							edge from parent[faulty]
						}
						child {node[response] {0,2}
							edge from parent[fault-free]
						}
						edge from parent[faulty]
					}
					child {node[response] {1,2}
						edge from parent[fault-free]
					}
					edge from parent[faulty]
				}
				child {node {-5,10}
					child {node[response] {3,0}
						edge from parent[faulty]
					}
					child {node[response] {2,0}
						edge from parent[fault-free]
					}
					edge from parent[fault-free]
				}
				edge from parent[faulty]
			}
			child {node[camino node] {8,-6}
				child {node {6,8}
					child {node {10,5}
						child {node[response] {3,0}
							edge from parent[faulty]
						}
						child {node[response] {2,0}
							edge from parent[fault-free]
						}
						edge from parent[faulty]
					}
					child {node[response] {2,-1}
						edge from parent[fault-free]
					}
					edge from parent[faulty]
				}
				child {node[camino node] {6,8}
					child {node[camino node] {10,5}
						child {node[response,camino node] {2,1}
							edge from parent[faulty,camino link]
						}
						child {node[response] {1,1}
							edge from parent[fault-free]
						}
						edge from parent[faulty,camino link]
					}
					child {node[response] {1,0}
						edge from parent[fault-free]
					}
					edge from parent[fault-free,camino link]
				}
				edge from parent[fault-free,camino link]
			}
			edge from parent[faulty,camino link]
		}
		child {node{-11,2}
			child {node{-6,-8}
				child {node {-8,6}
					child {node {-5,10}
						child {node[response] {-3,0}
							edge from parent[faulty]
						}
						child {node[response] {-2,0}
							edge from parent[fault-free]
						}
						edge from parent[faulty]
					}
					child {node[response] {-2,-1}
						edge from parent[fault-free]
					}
					edge from parent[faulty]
				}
				child {node{-8,6}
					child {node {-5,10}
						child {node[response] {-2,1}
							edge from parent[faulty]
						}
						child {node[response] {-1,1}
							edge from parent[fault-free]
						}
						edge from parent[faulty]
					}
					child {node[response] {-1,0}
						edge from parent[fault-free]
					}
					edge from parent[fault-free]
				}
				edge from parent[faulty]
			}
			child {node{-6,-8}
				child {node {8,-6}
					child {node {-10,-5}
						child {node[response] {0,-3}
							edge from parent[faulty]
						}
						child {node[response] {0,-2}
							edge from parent[fault-free]
						}
						edge from parent[faulty]
					}
					child {node {-10,-5}
						child {node[response] {-1,-2}
							edge from parent[faulty]
						}
						child {node[response] {-1,-1}
							edge from parent[fault-free]
						}
						edge from parent[fault-free]
					}
					edge from parent[faulty]
				}
				child {node{8,-6}
					child {node {5,-10}
						child {node[response] {1,-2}
							edge from parent[faulty]
						}
						child {node[response] {1,-1}
							edge from parent[fault-free]
						}
						edge from parent[faulty]
					}
					child {node {-2,-11}
						child {node[response] {0,-1}
							edge from parent[faulty]
						}
						child {node[response] {0,0}
					 		edge from parent[fault-free]
						}
						edge from parent[fault-free]
					}
					edge from parent[fault-free]
				}
				edge from parent[fault-free]
			}
			edge from parent[fault-free]
		}
	;
	\end{scope}
	\end{tikzpicture}
	\end{center}
	\caption{Binary decision tree for identification.}
	\label{fig:tree}
}
\end{figure}

\section{Conclusions}

In this paper a method for constructing a wide family of identifying codes over Cayley graphs of degree four defined by means of Abelian groups has been presented. These graphs include many previously known topologies for interconnection networks such as tori, twisted tori \cite{Sequin}, Midimews \cite{Beivide} and double loops in general \cite{Fiolnetworks}. The codes are constructed by means of subgroups of the generating group of the Cayley graph. As a consequence, the codewords are regularly distributed over the lattice so that the covering radius is relatively small. The identifying radius is characterized in terms of the generators of the subgroup as well as the covering radius. Moreover, some of the codes not only are identifying but also perfect. Therefore, the diagnosis process can be separated into two stages: a first stage for detecting the existence of a failure (using the small covering radius) and the identification of the corresponding vertex (using the identifying radius). Taking advantage of these facts, the adaptive identification process has been considered and a practical example of how it would perform has been addressed in the present paper. This special feature of the codes which are obtained with our method would make possible to assert that the construction gives suitable codes for a practical scenario.

\section*{Acknowledgments}
This work has been supported by the Spanish Ministry of Science under contracts
TIN2010-21291-C02-02, AP2010-4900 and CONSOLIDER Project CSD2007-00050, and by
the European HiPEAC Network of Excellence.

\bibliographystyle{abbrv}
\bibliography{main}


\end{document}